\documentclass[journal]{IEEEtran}
\usepackage{subfigure}
\usepackage{setspace}
\usepackage{amsmath}
\usepackage{amssymb}
\usepackage{amsfonts}
\usepackage{amscd}
\usepackage{mathrsfs}
\usepackage[final]{graphicx}
\usepackage{graphicx}
\usepackage{psfrag}
\usepackage{epsfig}
\usepackage{color}
\usepackage{url}
\usepackage{textcomp}
\usepackage{multirow}
\usepackage{threeparttable}
\input{epsf.sty}
\newtheorem{theorem}{Theorem}

\newtheorem{definition}{Definition}
\newtheorem{proposition}{Proposition}
\newtheorem{corollary}{Corollary}

\begin{document}

\title{Improved Perfect Space-Time Block Codes}
\author{
\authorblockN{K. Pavan Srinath and B. Sundar Rajan, {\it Senior Member, IEEE}\\}
\authorblockA{Dept of ECE, The Indian Institute of Science, \\
Bangalore-560012, India\\
Email:\{pavan,bsrajan\}@ece.iisc.ernet.in\\
}
}
\maketitle
\vspace{-15mm}
\begin{abstract}
Perfect space-time block codes (STBCs) are based on four design criteria - full-rateness, non-vanishing determinant, cubic shaping and uniform average transmitted energy per antenna per time slot. Cubic shaping and transmission at uniform average energy per antenna per time slot are important from the perspective of energy efficiency of STBCs. The shaping criterion demands that the {\it generator matrix} of the lattice from which each layer of the perfect STBC is carved be unitary. In this paper, it is shown that unitariness is not a necessary requirement for energy efficiency in the context of space-time coding with finite input constellations, and an alternative criterion is provided that enables one to obtain full-rate (rate of $n_t$ complex symbols per channel use for an $n_t$ transmit antenna system) STBCs with larger {\it normalized minimum determinants} than the perfect STBCs. Further, two such STBCs, one each for $4$ and $6$ transmit antennas, are presented and they are shown to have larger normalized minimum determinants than the comparable perfect STBCs which hitherto had the best known normalized minimum determinants.
\end{abstract}


 \section{Introduction and Background}\label{sec_intro}
  Perfect space-time block codes (STBCs) for multiple input, multiple output antenna (MIMO) systems were introduced in the landmark paper \cite{ORBV} for $2$, $3$, $4$ and $6$ transmit antennas. These were designed to meet four important criteria, namely
\begin{enumerate}
 \item full-rateness of STBCs.
 \item non-vanishing determinant (NVD) (see Definition \ref{nvd_def}).
 \item constellation cubic shaping (see subsection \ref{subsec_perfect}).
 \item uniform average transmitted energy per antenna per time slot.
\end{enumerate}
The first two criteria were shown to be sufficient for diversity-multiplexing gain tradeoff (DMT)-optimality and approximate universality \cite{elia}. The last two criteria were framed from the perspective of energy efficiency and hence coding gain. Later, perfect STBCs were constructed for arbitrary number of transmit antennas in \cite{new_per}. The perfect STBCs in general have among the largest known normalized minimum determinants (see Definition \ref{min_det_stbc}) among existing STBCs in their comparable class and in particular, the perfect STBCs of \cite{ORBV} have the largest known normalized minimum determinants for $2$, $3$, $4$ and $6$ transmit antennas. However, we note that the cubic shaping criterion, which demands that the generator matrix of each layer \cite{ORBV} of the codeword matrices of perfect STBCs be unitary, is not a necessary criterion (although sufficient) for energy efficiency in the context of space-time coding with finite input constellations. We propose an alternative criterion that preserves energy-efficiency and enables one to obtain STBCs with larger normalized minimum determinants than the perfect STBCs of \cite{ORBV} while meeting the other three design criteria. We then show the existence of one such STBC in literature for $4$ transmit antennas which has the best normalized minimum determinant. This STBC was first proposed in \cite{spcom} but its superior coding gain was not identified. We then present a new STBC for $6$ transmit antennas which, to the best of our knowledge, has the largest normalized minimum determinant for $6$ transmit antennas. We call these STBCs ``improved perfect STBCs'' (see Definition \ref{imp_per} in Section \ref{sec_mod_shaping}). 

\subsection{Contributions and paper organization}
The contributions of this paper may be summarized as follows.
\begin{enumerate}
\item We propose a modified shaping criterion that enables one to obtain rate-$n_t$ STBCs with larger coding gains than the perfect STBCs while retaining all the other desirable properties of the perfect STBCs. 
 \item For $4$ and $6$ transmit antennas, we present such STBCs which have a larger normalized minimum determinant than the comparable perfect STBCs.
\end{enumerate}

The paper is organized as follows. In Section \ref{sec_system_model}, we give the system model, relevant definitions and a brief overview of perfect STBCs. Section \ref{sec_mod_shaping} presents the modified shaping criterion while the improved perfect STBCs for $4$ and $6$ transmit antennas are presented in Sections \ref{stbc_4} and \ref{stbc_6}, respectively. Appendix I provides some basic definitions and results in number theory which are used in this paper.

\subsection*{Notations}
Throughout the paper, the following notations are used. 
\begin{itemize}
 \item Bold, lowercase letters denote vectors, and bold, uppercase letters denote matrices.
 \item $\textbf{X}^{H}$, $\textbf{X}^{T}$, $det(\textbf{X})$, $tr(\textbf{X})$ and $\Vert \textbf{X} \Vert$ denote the conjugate transpose, the transpose, the determinant, the trace and the Frobenius norm of $\textbf{X}$, respectively. 
\item $\vert \mathcal{S}\vert$ denotes the cardinality of the set $\mathcal{S}$ and for the set $\mathcal{T} \subset \mathcal{S}$, $\mathcal{S} \setminus  \mathcal{T}$ denotes the set of elements of $\mathcal{S}$ not in $\mathcal{T}$. 
\item $\textbf{I}$ and $\textbf{O}$ denote the identity matrix and the null matrix of appropriate dimensions.
\item $\mathbb{E}(X)$ denotes the expectation of the random variable $X$. 
\item $\mathbb{R}$, $\mathbb{C}$ and $\mathbb{Q}$ denote the field of real, complex and rational numbers, respectively. $\mathbb{Z}$ denotes the ring of rational integers. 
\item Unless used as an index, a subscript or a superscript, $i$ denotes $\sqrt{-1}$ and $\omega$ denotes the primitive third root of unity. 
\item For fields $\mathbb{K}$ and $\mathbb{F}$, $\mathbb{K}/\mathbb{F}$ denotes that $\mathbb{K}$ is an extension of $\mathbb{F}$ and $[\mathbb{K}:\mathbb{F} ] = m$ indicates that $\mathbb{K}$ is a finite extension of $\mathbb{F}$ of degree $m$. 
\item$Gal( \mathbb{K}/\mathbb{F})$ denotes the Galois group of $\mathbb{K}/\mathbb{F}$, i.e., the group of $\mathbb{F}$-linear automorphisms of $\mathbb{K}$. 
\item For an element $a$ of a ring $\mathcal{R}$, $a\mathcal{R}$ denotes the ideal of $\mathcal{R}$ generated by $a$.   
\end{itemize}

\section{System Model and definitions}\label{sec_system_model}
We consider an $n_t$ transmit antenna, $n_r$ receive antenna MIMO system ($n_t\times n_r$ system) with perfect channel-state information available at the receiver (CSIR) alone. The channel is assumed to be quasi-static with Rayleigh fading. The system model is 
\begin{equation}\label{ch_mod}	
 \textbf{Y} = \sqrt{\rho}\textbf{HS} + \textbf{N}
\end{equation}

\noindent where $\textbf{Y} \in \mathbb{C}^{n_r\times {\textrm{T}}}$ is the received signal matrix, $\textbf{S} \in \mathbb{C}^{n_t\times {\textrm{T}}}$ is the codeword matrix that is transmitted over a block of ${\textrm{T}}$ channel uses, $\textbf{H} \in \mathbb{C}^{n_r\times n_t}$ and $\textbf{N} \in \mathbb{C}^{n_r\times {\textrm{T}}}$ are respectively the channel matrix and the noise matrix with entries independently and identically distributed (i.i.d.) circularly symmetric complex Gaussian random variables with zero mean and unit variance. The average signal-to-noise ratio (SNR) at each receive antenna is denoted by $\rho$. It follows that 
\begin{equation}\label{energy_con}
 \mathbb{E}(\Vert \textbf{S} \Vert^2) = {\textrm{T}} .
\end{equation}

A space-time block code (STBC) $\mathcal{S}$ of block-length ${\textrm{T}}$ for an $n_t$ transmit antenna MIMO system is a finite set of complex matrices of size $n_t \times {\textrm{T}}$. An STBC transmitting $k$ independent complex information symbols in ${\textrm{T}}$ channel uses is said to have a rate of $k/{\textrm{T}}$ complex symbols per channel use. Throughout the paper, we consider linear STBCs \cite{HaH} whose codeword matrices are of the form $\textbf{S} = \sum_{i=1}^{k}s_i\textbf{A}_i$ where the $k$ independent information symbols $s_i$ take values from a complex constellation $\mathcal{A}_q$ which is QAM or HEX, and $\textbf{A}_i$, $i=1,\cdots,k$, are the complex weight matrices of the STBC. An $M$-PAM, $M$-QAM and $M$-HEX with $M = 2^a$, $a$ even and positive, are respectively given as  
\begin{eqnarray*}
 M\textrm{-PAM} & = & \{ -M+1,-M+3, -M+5,\cdots,M-1\},\\
 M\textrm{-QAM} & = & \left\{ a + ib, a,b \in \sqrt{M}\textrm{-PAM}\right\},\\
 M\textrm{-HEX} & = & \left\{ a + \omega b, a,b \in \sqrt{M}\textrm{-PAM}\right\}.
\end{eqnarray*}
Among STBCs transmitting at the same rate in bits per channel use, the metric for comparison that decides their error performance is the normalized minimum determinant which is defined as follows.

\begin{figure*}
\begin{equation}\label{form_div}
 \textbf{F} = \left[ \begin{array}{ccccc}
                      a_0 & \gamma\tau(a_{n-1}) & \gamma \tau^2(a_{n-2}) & \cdots & \gamma\tau^{n-1}(a_{1})\\
                      a_1 & \tau(a_0) & \gamma\tau^2(a_{n-1}) & \cdots & \gamma\tau^{n-1}(a_2)\\
                      a_2 & \tau(a_1) & \tau^2(a_0) & \cdots & \gamma\tau^{n-1}(a_3)\\
                      a_3 & \tau(a_2) & \tau^2(a_1) & \cdots & \gamma\tau^{n-1}(a_4)\\
                      \vdots & \vdots & \vdots &  & \vdots\\
                      a_{n-1} & \tau(a_{n-2}) & \tau^2(a_{n-3}) & \cdots & \tau^{n-1}(a_0)\\                      
                     \end{array} \right].
\end{equation}
\hrule
\end{figure*}

\begin{definition}\label{min_det_stbc}({\it Normalized minimum determinant})
 For an STBC $\mathcal{S}$ whose codeword matrices satisfy \eqref{energy_con}, the normalized minimum determinant $\delta_{min}(\mathcal{S})$ is defined as
\begin{equation}\label{min_det_eq}
\delta_{min}(\mathcal{S}) = \min_{\textbf{S}_i,\textbf{S}_j \in \mathcal{S},i\neq j} \left\{ \left \vert det\left(\textbf{S}_i - \textbf{S}_j\right) \right\vert^2 \right\}.  
\end{equation}
\end{definition}
For full-diversity STBCs, $\delta_{min}(\mathcal{S})$ defines the coding gain \cite{TSC}. Between two competing STBCs with the same rate in bits per channel use, the one with the larger normalized minimum determinant is expected to have a better error performance. 

\begin{definition}\label{scheme} ({\it STBC-scheme} \cite{tse}) An STBC-scheme $\mathcal{S}_{sch}$ is defined as a family of STBCs indexed by $\rho$, each STBC of block length $\textrm{T}$ so that $\mathcal{S}_{sch} = \{ \mathcal{S}(\rho)\}$, where the STBC $\mathcal{S}(\rho)$ corresponds to an average signal-to-noise ratio of $\rho$ at each receive antenna.  
\end{definition}

\begin{definition}\label{nvd_def}({\it Non-vanishing determinant} \cite{BRV}) A linear STBC-scheme $\mathcal{S}_{sch} = \{\mathcal{S}(\rho)\}$, all of whose STBCs $\mathcal{S}(\rho)$ are defined by weight matrices $\{ \textbf{A}_{i},i=1,\cdots,k \}$ and employ complex constellations (QAM or HEX) that are finite subsets of an infinite complex lattice $\mathcal{A}_L$ ($\mathbb{Z}[i]$ or $\mathbb{Z}[\omega]$), is said to have the non-vanishing determinant (NVD) property if $\mathcal{S}_{\infty}  \triangleq \left\{ \sum_{i=1}^k s_{i}\textbf{A}_{i}  \vert s_{i} \in \mathcal{A}_L \right\}$ is such that 
\begin{equation*}
 \min_{\textbf{S} \in \mathcal{S}_{\infty}, \textbf{S} \neq \textbf{O}} \left\{ \vert det(\textbf{S})  \vert^2\right\} = c > 0
\end{equation*}
for some strictly positive constant $c$.
\end{definition}

\begin{definition}\label{gen_mat}({\it Generator matrix of an STBC}) For a linear STBC that is given by $\mathcal{S} = \left\{ \sum_{i=1}^k s_i\textbf{A}_i \right\}$, the {\emph{generator}} matrix $\textbf{G} \in \mathbb{C}^{{\textrm{T}}n_t \times k}$ is defined as \cite{HaH}
\begin{equation*}
\textbf{G} = \left[vec(\textbf{A}_{1})\ vec(\textbf{A}_{2})\ \cdots \ vec(\textbf{A}_{k})\right]
\end{equation*}
where the operation $vec(\textbf{A})$ denotes the vector obtained by stacking the columns of $\textbf{A}$ one below the other.
\end{definition}

\subsection{Cyclic Division Algebras}
A cyclic division algebra (CDA) $\mathcal{A}$ of degree $n$ over a number field $\mathbb{F}$ is a vector space over $\mathbb{F}$ of dimension $n^2$. The center of $\mathcal{A}$ is $\mathbb{F}$ and there exists a maximal subfield $\mathbb{K}$ of $\mathcal{A}$ such that $\mathbb{K}$ is a Galois extension of degree $n$ over $\mathbb{F}$ with a cyclic Galois group generated by $\tau$. $\mathcal{A}$ is a right vector space over $\mathbb{K}$ and can be expressed as 
\begin{equation*}\label{cda_def}
 \mathcal{A} = \mathbb{K} \oplus \textbf{i}\mathbb{K} \oplus \textbf{i}^2\mathbb{K} \oplus \cdots \oplus \textbf{i}^{n-1}\mathbb{K}
\end{equation*}
 where $a\textbf{i}  =  \textbf{i}\tau(a)$, $\forall a \in \mathbb{K}$, $\textbf{i}^n  =  \gamma$ for some $\gamma \in \mathbb{F}^{\times} = \mathbb{F}\setminus \{0\}$ such that the norm $N_{\mathbb{K}/\mathbb{F}}(a) = \prod_{i=0}^{n-1}\tau^i(a)$ of any element $a \in \mathbb{K}$ satisfies  
\begin{equation}\label{non_norm}
 N_{\mathbb{K}/\mathbb{F}}(a) \neq \gamma^t, ~ t= 1,\cdots,n-1.
\end{equation}

The CDA $\mathcal{A}$ is denoted by $(\mathbb{K}/\mathbb{F}, \tau, \gamma)$. $\mathcal{A}$ has a matrix representation and in particular, an element $a_0+\textbf{i}a_1+\cdots+\textbf{i}^{n-1}a_{n-1} $ of $\mathcal{A}$, where $a_i \in \mathbb{K}$, has the representation shown in \eqref{form_div} at the top of the next page. In addition, every nonzero matrix of the form shown in \eqref{form_div} is invertible and its determinant lies in $\mathbb{F}^{\times} $ \cite{jacobson}, i.e.,
\begin{equation}\label{det_f}
 det(\textbf{F}) \in \mathbb{F}^{\times},  ~~\textbf{F} \neq \textbf{O}.
\end{equation}

 For more on CDAs, one can refer to \cite{jacobson}, \cite{sethuraman}, and references therein. 

\subsection{STBCs from CDA}\label{subsec_cda}
 For the purpose of space-time coding, the signal constellation is generally $M$-QAM or $M$-HEX which are finite subsets of $\mathbb{Z}[i]$ and $\mathbb{Z}[\omega]$, respectively. So, $\mathbb{F}$ is naturally chosen to be $\mathbb{Q}(i)$ or $\mathbb{Q}(\omega)$ for which the ring of integers are respectively $\mathbb{Z}[i]$ and $\mathbb{Z}[\omega]$, and a CDA $\mathcal{A}$ of degree $n_t$ over $\mathbb{F}$ is constructed. We denote the ring of integers of $\mathbb{F}$ and $\mathbb{K}$ by $\mathcal{O}_\mathbb{F}$ and $\mathcal{O}_\mathbb{K}$, respectively. The codeword matrices of the STBC obtained from the CDA $\mathcal{A}$ have the structure shown in \eqref{form_div} with $a_i$, $i=0,1,\cdots,n_t-1$, expressed as linear combinations of elements of some chosen $\mathbb{F}$-basis over $\mathcal{O}_\mathbb{F}$, and hence STBCs from CDAs encode $n_t^2$ complex information symbols in $n_t$ channel uses. An STBC $\mathcal{S}$ that is obtained from CDA is expressible (prior to SNR normalization) as $\mathcal{S} = \left\{\sum_{i=1}^{n_t^2}s_i \textbf{A}_i, s_i \in \mathcal{A}_q\right\}$ where $\mathcal{A}_q$ is either QAM or HEX, and $\textbf{A}_i$, $1,\cdots,n_t$, are the complex weight matrices. The following proposition relates the choice of $\mathbb{F}$-basis to the NVD property of STBC-schemes that are based on STBCs from CDA.

\begin{proposition}\label{prop1}
 An STBC-scheme that is based on STBCs from CDA has a non-vanishing determinant if all the elements of the $\mathbb{F}$-basis belong to $\mathcal{O}_\mathbb{K}$.
\end{proposition}
\begin{proof}
 Consider the STBC-scheme $\mathcal{S}_{sch} = \{\mathcal{S}(\rho)\}$, where all the $\mathcal{S}(\rho)$ are obtained from the same CDA and given by $\mathcal{S}(\rho) = \{ \beta\sum_{i=1}^{n_t^2}s_i \textbf{A}_i, s_i \in \mathcal{A}_q(\rho)\}$, where $\mathcal{A}_q(\rho)$ is the regular QAM or HEX constellation whose size is dependent on $\rho$ so that the required multiplexing gain is achieved (see \cite{elia} for details), and $\beta$ is the normalizing scalar that ensures that the average SNR at each receive antenna is $\rho$. From Definition \ref{nvd_def}, $\mathcal{S}_{sch}$ has the NVD property if $\mathcal{S}_{\infty} = \left\{ \sum_{i=1}^{n_t^2}s_i \textbf{A}_i, s_i \in \mathcal{O}_\mathbb{F}\right\}$ ($\mathcal{O}_\mathbb{F}$ is either $\mathbb{Z}[i]$ or $\mathbb{Z}[\omega]$) is such that 
\begin{equation*}
 \min_{\textbf{S} \in \mathcal{S}_{\infty}, \textbf{S} \neq \textbf{O}} \left\{ \vert det(\textbf{S})  \vert^2\right\} = c > 0
\end{equation*}
for some constant $c$. Let the $\mathbb{F}$-basis $\{\theta_i,i=1,\cdots,n_t \}$ be such that all the $\theta_i$ belong to $\mathcal{O}_\mathbb{K}$. Since $\gamma \in \mathbb{F}$ and satisfies \eqref{non_norm}, we can express $\gamma$  as $\gamma = \frac{a}{b}$ with $a,b \in \mathcal{O}_{\mathbb{F}} \setminus \{0\}$. Now, multiplying all the matrices of $\mathcal{S}_{\infty}$ by $b$ results in all the entries of all the matrices of $\mathcal{S}_{\infty}$ belonging to $\mathcal{O}_{\mathbb{K}}$ and from \eqref{det_f}, any nonzero matrix of $\mathcal{S}_{\infty}$ has a determinant that belongs to $(\mathbb{F} \cap \mathcal{O}_\mathbb{K}) \setminus \{0\} = \mathcal{O}_\mathbb{F} \setminus \{0\} $. Since $\mathcal{O}_\mathbb{F}$ is either $\mathbb{Z}[i]$ or $\mathbb{Z}[\omega]$, we have
\begin{equation*}
 \min_{\textbf{S} \in \mathcal{S}_{\infty}, \textbf{S} \neq \textbf{O}} \left\{ \vert det(\textbf{S})  \vert^2\right\} \geq \frac{1}{\vert b \vert^{2n_t}} > 0
\end{equation*}
which proves the proposition.
\end{proof}

So, for the purpose of space-time coding, an $\mathbb{F}$-basis $\{\theta_i, i=1,2,\cdots,n_t ~\vert~ \theta_i \in \mathcal{O}_\mathbb{K}\}$ is chosen (this can also be an $\mathcal{O}_\mathbb{F}$-basis of $\mathcal{O}_\mathbb{K}$) and the $a_i \in \mathbb{K}$ in \eqref{form_div} are expressed as linear combinations of elements of this basis over $\mathcal{O}_\mathbb{F}$. The STBC which encodes symbols from a complex constellation $\mathcal{A}_q$ ($M$-QAM or $M$-HEX) has its codewords of the form shown in \eqref{form_div} with $a_i = \sum_{j=1}^{n_t}s_{ij}\theta_j$, $s_{ij} \in \mathcal{A}_q \subset \mathcal{O}_{\mathbb{F}}$ with $\mathcal{O}_{\mathbb{F}} = \mathbb{Z}[i]$ or $\mathbb{Z}[\omega]$. A codeword matrix of STBCs from CDA has $n_t$ layers \cite{ORBV}, with the $(i+1)^{th}$ layer transmitting the vector $\textbf{D}_{i}[a_{i}, \tau(a_{i}), \cdots, \tau^{n_t-1}(a_{i}) ]^T$, $i =0,\cdots,n_t-1$, where $\textbf{D}_i$ is a diagonal matrix given by
 \begin{equation}\label{form_Di}
  \textbf{D}_{i} \triangleq \textrm{diag}[\underbrace{1,\cdots,1}_{n_t-i \textrm{ times}}, \underbrace{\gamma,\cdots,\gamma}_{i\textrm{ times}}]
 \end{equation}
and $[a_{i}, \tau(a_{i}), \cdots, \tau^{n_t-1}(a_{i}) ]^T =\textbf{R}\textbf{s}_i, i =0,\cdots,n_t-1, $ where $\textbf{s}_i = [s_{i1}, s_{i2},\cdots,s_{in_t}]^T  \in \mathcal{A}_q^{n_t\times1}$ and $\textbf{R} \in \mathbb{C}^{n_t\times n_t}$ is the {\it generator matrix} of each layer of the STBC (not to be confused with the generator matrix $\textbf{G}$ of the STBC which is given by Definition \ref{gen_mat}) and is given as \begin{equation}\label{integral_basis}
 \textbf{R} = \frac{1}{\sqrt{\lambda}}\left[ \begin{array}{ccc}
                      \theta_1 &  \cdots & \theta_{n_t}\\
                      \tau(\theta_1)  &  \cdots & \tau(\theta_{n_t})\\
                      \vdots &  \vdots  & \vdots \\
                      \tau^{n_t-1}(\theta_1)  &  \cdots & \tau^{n_t-1}(\theta_{n_t})\\                     
                     \end{array} \right]
\end{equation}
where, as mentioned earlier, $\{ \theta_i, i = 1, 2,\cdots,n_t ~\vert~ \theta_i \in \mathcal{O}_\mathbb{K}\}$ is an ${\mathbb{F}}$-basis of ${\mathbb{K}}$ and $\lambda$ is a suitable real-valued scalar designed so that the STBC meets the energy constraint in \eqref{energy_con}.
\subsection{Perfect Codes}\label{subsec_perfect}

 The perfect STBCs are designed to be equipped with the following two desirable properties \cite{ORBV}, \cite{new_per}. 
\begin{itemize}
 \item {\it Approximate-universality} : This is achieved if the STBC satisfies the following criteria.
\begin{description}
 \item [C1] {\it Full-rate}\footnote{In this paper, a rate-$n_t$ STBC is referred to as a full-rate STBC.} : The STBC transmits $n_t^2$ independent complex information symbols in $n_t$ channel uses. 
\item [C2] {\it Non-vanishing determinant} : The STBC-scheme has the NVD property. 
\end{description}
\item {\it Energy-efficiency/coding gain} : To achieve this, the STBC should satisfy the following criteria.
\begin{description}
\item [C3] {\it Constellation shaping criterion} : The matrix $\textbf{R}$ given by \eqref{integral_basis} is unitary \cite{ORBV} so that on each layer, the energy required to transmit the linear combination of information symbols is equal to the energy required to transmit the information symbols themselves, i.e., $\Vert\textbf{Rs}_i\Vert^2 = \Vert \textbf{s}_i\Vert^2$, $i=0,\cdots,n_t-1$, with the notations as used in the previous subsection.
\item [C4] {\it Uniform average transmitted energy} : The average transmitted energy for all the antennas in all time slots is the same.
\end{description}
\end{itemize}

To satisfy C1, $\mathbb{F}$ is chosen to be $ \mathbb{Q}(i)$ or $\mathbb{Q}(\omega)$ and  a CDA of degree $n_t$ over $\mathbb{F}$ is constructed. C2 is satisfied by choosing an $\mathbb{F}$-basis $\{\theta_i, i=1,2,\cdots,n_t ~\vert~ \theta_i \in \mathcal{O}_{\mathbb{K}}\}$ which guarantees a non-vanishing determinant from Proposition \ref{prop1}.

 C3 is satisfied by choosing the $\mathbb{F}$-basis $\{\theta_i, i=1,2,\cdots,n_t ~\vert~ \theta_i \in \mathcal{O}_{\mathbb{K}}\}$ such that $\textbf{R}$ 
is unitary. C4 is satisfied by choosing $\gamma$ such that $\vert \gamma \vert^2 =1 $. In \cite{ORBV}, $\gamma$ is chosen to be in $\mathcal{O}_\mathbb{F}$ while in \cite{new_per}, $\gamma$ is chosen to be the ratio of a suitable element $a \in \mathcal{O}_\mathbb{F} \setminus \{0\}$ and its complex conjugate. In the former case, the minimum determinant, prior to normalizaton, is a nonzero positive integer while in the latter case, it is $\frac{1}{\vert a \vert^{2(n_t-1)}}$ \cite{new_per}. Choosing $\gamma$ to be in $\mathcal{O}_\mathbb{F}$ restricts the construction of the perfect STBCs to only $2,3,4$ and $6$ transmit antennas \cite{ORBV} but these STBCs have the largest known coding gains in their class\footnote{There are certain non-linear STBCs, for example in \cite{max_order}, which beat the Golden code. These STBCs employ spherical shaping, involve additional complexity in encoding and are not sphere-decodable. We do not consider this class of non-linear STBCs in this paper.}.

\section{Modified Shaping criterion}\label{sec_mod_shaping}

For an STBC that is obtained from CDA to be energy efficient, C3, which asks for $\textbf{R}$ to be unitary, is a sufficient but not a necessary criterion - it is not necessary that on the $i^{th}$ layer, the energy required to transmit $a_{i-1}$, $\tau(a_{i-1})$, $\cdots$, and $\tau^{n_t-1}(a_{i-1})$ be equal to the energy used for sending the information symbols $s_{ij}$ themselves. It is sufficient that the {\it average} energy required to send the linear combination of the information symbols on each layer is equal to the {\it average} energy used for sending the information symbols themselves, i.e., $\mathbb{E}\left(\Vert\textbf{Rs}_i\Vert^2\right) = \mathbb{E}\left(\Vert \textbf{s}_i\Vert^2\right)$, $i=0,\cdots,n_t-1$ (with the notations as in Subsection \ref{subsec_cda}), where the expectation is over the distribution of $\textbf{s}_i$ which by assumption has probability mass function (PMF) given by $p_{\textbf{s}_i}(\textbf{s}) = (1/M)^{n_t}$, $\forall\textbf{s} \in \mathcal{A}_q^{n_t\times1}$. Hence, unitariness of $\textbf{R}$ is not necessary. However, in literature, a unitary $\textbf{R}$ is seen as desirable as it makes the STBC information-lossless. We elaborate on this in the following subsection.

\subsection{Unitary generator matrix $\textbf{G}$ and information-losslessness}
An STBC is said to be {\it information-lossless} \cite{damen} if the maximum instantaneous mutual information of the equivalent MIMO channel after space-time processing is the same as the maximum instantaneous mutual information of the MIMO channel without space-time processing. The maximum instantaneous mutual information (in bits per channel use) supported by the MIMO channel without an STBC encoder is \cite{tel}  
\begin{equation*}
 C(\textbf{H}) = \max_{tr(\textbf{Q}) \leq \rho}\log_2det\left(\textbf{I}+\textbf{HQ}\textbf{H}^H \right)
\end{equation*}
where $\textbf{Q}$ is a non-negative definite matrix. A good approximation for $\textbf{Q}$ is taken\footnote{For calculating the ergodic capacity which is the expectation of $C(\textbf{H})$ over the distribution of $\textbf{H}$, $(\rho/n_t)\textbf{I}$ is the optimal $\textbf{Q}$.} to be $(\rho/n_t)\textbf{I}$ so that 
\begin{equation}\label{cap1}
 C(\textbf{H}) \approx \log_2det\left(\textbf{I}+\frac{\rho}{n_t}\textbf{H}\textbf{H}^H \right).
\end{equation} 
Now, for linear STBCs of the form $\mathcal{S} = \{ \sum_{i=1}^{k}s_i\textbf{A}_i\}$, the signal model given in \eqref{ch_mod} can be rewritten as
\begin{equation*}
vec(\textbf{Y}) = \sqrt{\rho} (\textbf{I}_\textrm{T} \otimes \textbf{H})\textbf{G}\textbf{s} + vec(\textbf{N}) 
\end{equation*}
where $\textbf{I}_\textrm{T}$ is the identity matrix of size $\textrm{T} \times \textrm{T}$, $\textbf{G}$ is the generator matrix defined in Definition \ref{gen_mat} and $\textbf{s}$ is the vector of information symbols belonging to $\mathcal{A}_q^{k\times1}$. For this model, the maximum mutual information for a given channel matrix $\textbf{H}$ is 
\begin{equation*}
 C^{\prime}(\textbf{H}) = \max_{tr(\textbf{GQ}^\prime\textbf{G}^H )\leq \rho\textrm{T}  }\left(\frac{1}{\textrm{T}}\log_2det\left(\textbf{I}+ \bar{\textbf{H}}G\textbf{Q}^\prime\textbf{G}^H\bar{\textbf{H}}^H \right)\right)
\end{equation*} 
where $\textbf{Q}^\prime$ is non-negative definite and $\bar{\textbf{H}} = \textbf{I}_\textrm{T} \otimes \textbf{H}$. When $\textbf{G}$ is unitary (possible only when $k = n_t \textrm{T}$) and $\textbf{Q}^\prime = (\rho/n_t)\textbf{I}$, $C^\prime(\textbf{H})$ is equal to $C(\textbf{H})$ (assuming $C(\textbf{H})$ is equal to the right hand side of \eqref{cap1}) and hence the STBC is information-lossless \cite{HaH}, \cite{damen}. For STBCs from CDA, if $\textbf{R}$ given by \eqref{integral_basis} is unitary, so is $\textbf{G}$.

However, it is important to note that the expressions for both $C(\textbf{H})$ and $C^\prime(\textbf{H})$ are obtained for {\it Gaussian inputs} (since the entropy of the output is maximized if and only if the input is Gaussian). In the case of STBCs, the inputs information symbols take values from $\mathcal{A}_q$ which is $M$-QAM or $M$-HEX, and all the signal points are equally likely to be chosen so that the PMF of $s_i$ is $p_{s_i}(s) = 1/M$, $\forall s\in \mathcal{A}_q$. So, for the signal model $\textbf{y} = \sqrt{\beta} \textbf{Hs} + \textbf{n}$ where $\textbf{s} \in \mathcal{A}_q^{n_t \times 1}$ and $\mathbb{E}\left(\Vert\sqrt{\beta}\textbf{Hs}\Vert^2\right) = \rho n_r$, the constellation constrained mutual information $C_c(\textbf{H})$ is not given by \eqref{cap1} but by the following expression \cite{hoch}, \cite{bar}.
\begin{eqnarray}\nonumber 
 C_{c}(\textbf{H}) & = & -\mathbb{E}\log_2\left( \frac{1}{(M\pi)^{n_t}}\sum_{\textbf{s} \in \mathcal{A}_q^{n_t\times1}}e^{-\Vert \textbf{y} - \sqrt{\beta}\textbf{Hs} \Vert^2} \right) \\
\label{cap_con1}
& & - n_t\log_2(\pi e)
\end{eqnarray}
where the expectation is over the distribution of $\textbf{y}$. With space time coding, the corresponding constellation constrained mutual information is 
 \begin{eqnarray}\nonumber 
 C^\prime_{c}(\textbf{H}) & = & -\frac{1}{\textrm{T}}\mathbb{E}\log_2\left( \frac{1}{(M\pi)^{n_t\textrm{T}}}\sum_{\textbf{s} \in \mathcal{A}_q^{n_t\textrm{T}\times1}}e^{-\Vert \textbf{y}^\prime - \sqrt{\rho}\bar{\textbf{H}}\textbf{Gs} \Vert^2} \right) \\
\label{cap_con2}
& & - n_t\log_2(\pi e)
\end{eqnarray}
where $\textbf{y}^\prime = vec(\textbf{Y})$ and the expectation is over the distribution of $\textbf{y}^\prime$, and $\mathbb{E}\left(\Vert\textbf{Gs}\Vert^2\right) = \textrm{T}$. It is clear from \eqref{cap_con1} and \eqref{cap_con2} that the significance of unitariness (or scaled unitariness) of the generator matrix $\textbf{G}$ is questionable when finite constellations are used. In particular, the notion of information-lossless STBCs is itself questionable.

\subsection{Modified Shaping criterion}

Having noted that unitariness of $\textbf{G}$ and hence of $\textbf{R}$ is not a necessary criterion, we propose a change in C3 as follows. The modified shaping criterion can be separated into two subcriteria which are 
\begin{description}
 \item[C3.1] the {\it average} energy required to transmit the linear combination of the information symbols on each layer is equal to the {\it average} energy used for sending the information symbols themselves, i.e., $\mathbb{E}\left(\Vert\textbf{Rs}_i\Vert^2\right) = \mathbb{E}\left(\Vert \textbf{s}_i\Vert^2\right)$, $i=0,\cdots,n_t-1$, where the expectation is over the distribution of $\textbf{s}_i$ which by assumption has a PMF given by $p_{\textbf{s}_i}(\textbf{s}) = (1/M)^{n_t}$, $\forall\textbf{s} \in \mathcal{A}_q^{n_t\times1}$.
\item[C3.2] All the $n_t^2$ symbols are transmitted at the same average energy.
\end{description}

The rationale behind C3.1 is obvious - we do not wish to blow up the average energy required to transmit the information symbols. The reason for coming up with C3.2 is that no symbol should be favoured over other symbols with respect to energy required for transmission. We assume that the average energy of $\mathcal{A}_Q$ is $E$ so that $\mathbb{E}( \Vert \textbf{s}_i\Vert^2) = n_tE$, and because of the symmetry of $M$-QAM and $M$-HEX, we have $\mathbb{E}( \textbf{s}_i\textbf{s}_i^H) =E\textbf{I}$. It is also assumed that $\vert \gamma \vert^2 = 1$ so that $\textbf{D}_{i}$ given by \eqref{form_Di} is unitary, since it is a necessary condition for C4 to be satisfied. With these assumptions, we have the following proposition.

\begin{proposition}
C3.1, C3.2 and C4 are together satisfied if and only if $\textbf{R}$ given by \eqref{integral_basis} is such that all of its rows and columns have a Euclidean norm equal to unity. 
\end{proposition}
\begin{proof}
We prove that if C3.1, C3.2 and C4 are together satisfied, then $\textbf{R}$ shall be such that all of its rows and columns have a Euclidean norm equal to unity. The converse is then easy to see. If C3.1 is satisfied, then, with $\textbf{D}_i$ unitary, we have
\begin{eqnarray}\nonumber
   \mathbb{E}( \left \Vert \textbf{s}_i\right\Vert^2)& =&\mathbb{E}\left(\left \Vert \textbf{Rs}_i\right\Vert^2  \right)=\mathbb{E}\left[ tr\left(\textbf{Rs}_i (\textbf{Rs}_i)^H\right)\right]\\
\nonumber
& = & \mathbb{E}\left[ tr\left(\textbf{Rs}_i\textbf{s}_i^H \textbf{R}^H\right)\right] = tr\left[\mathbb{E}\left( \textbf{Rs}_i\textbf{s}_i^H \textbf{R}^H\right)\right]\\ 
\nonumber
& = & tr\left[\textbf{R}\mathbb{E}\left(\textbf{s}_i\textbf{s}_i^H\right) \textbf{R}^H\right]=  tr\left[\textbf{R}(E\textbf{I}) \textbf{R}^H\right]\\
\label{eq12}
&=&  E\sum_{i=1}^{n_t}\Vert \textbf{r}_i \Vert^2
\end{eqnarray} 
where $\textbf{r}_i$ denotes the $i^{th}$ row of $\textbf{R}$. It follows that for C4 to be satisfied,
\begin{equation}\label{eq13}
 \mathbb{E}(\vert \textbf{r}_1\textbf{s}_i\vert^2)=\mathbb{E}(\vert \textbf{r}_2\textbf{s}_i\vert^2)=\cdots =\mathbb{E}(\vert \textbf{r}_{n_t}\textbf{s}_i\vert^2), 
\end{equation}
 $\forall i = 0,\cdots,n_t-1$. So, from \eqref{eq12}, \eqref{eq13} and the fact that $\mathbb{E}( \Vert \textbf{s}_i\Vert^2) = n_tE$, $\textbf{R}$ must satisfy $\Vert \textbf{r}_1 \Vert^2 = \Vert \textbf{r}_2 \Vert^2 = \cdots = \Vert \textbf{r}_{n_t} \Vert^2 = 1$. Now, denoting the $i^{th}$ column of $\textbf{R}$ by $\textbf{r}_i^\prime$, we have
\begin{eqnarray}\nonumber
   \mathbb{E}( \left \Vert \textbf{s}_i\right\Vert^2)& =&\mathbb{E}\left(\left \Vert \textbf{Rs}_i\right\Vert^2  \right) = \mathbb{E}\left[ \textbf{s}_i^H\textbf{R}^H\textbf{R}\textbf{s}_i\right]   \\
\nonumber
& = &  E\sum_{i=1}^{n_t}\Vert \textbf{r}_i^\prime \Vert^2.
\end{eqnarray} 
But C3.2 demands that $\Vert \textbf{r}_1^\prime \Vert^2=\Vert \textbf{r}_2^\prime \Vert^2=\cdots=\Vert \textbf{r}_{n_t}^\prime \Vert^2$. Hence, the Euclidean norm of each row and column of $\textbf{R}$ should be equal to unity. This concludes the proof.
\end{proof}

\begin{figure*}
\begin{equation}\label{c6}
\textbf{S} = \left[\begin{array}{rrrrrr}
                      a_0 & -\omega\tau(a_{5}) & -\omega \tau^2(a_{4}) & -\omega \tau^{3}(a_{3}) & -\omega\tau^4(a_{2}) &  -\omega\tau^5(a_{1})\\
                      a_1 & \tau(a_0) & -\omega\tau^2(a_{5}) & -\omega\tau^{3}(a_4) & -\omega\tau^4(a_{3}) &  -\omega\tau^5(a_{2})\\
                      a_2 & \tau(a_1) & \tau^2(a_0) & -\omega\tau^{3}(a_5) & -\omega\tau^4(a_{4}) &  -\omega\tau^5(a_{3})\\                      
                      a_{3} & \tau(a_{2}) & \tau^2(a_{1}) & \tau^{3}(a_0) & -\omega\tau^4(a_{5}) &  -\omega\tau^5(a_{4})\\ 
                      a_{4} & \tau(a_{3}) & \tau^2(a_{2}) & \tau^{3}(a_1) & \tau^4(a_{0}) &  -\omega\tau^5(a_{5}) \\ 
                      a_{5} & \tau(a_{4}) & \tau^2(a_{3}) & \tau^{3}(a_2) & \tau^4(a_{1}) &  \tau^5(a_{0})\\ 
                   \end{array}
\right] 
\end{equation}
\hrule
\end{figure*}

An STBC with a unitary $\textbf{R}$ obviously satisfies C3.1 and C3.2 but unitariness is not a necessary condition. In the following two sections, we highlight the significance of the modified shaping criterion by showing the existence of STBCs which do not have a unitary $\textbf{R}$ but have a higher coding gain than the perfect STBCs for $4$ and $6$ transmit antennas \cite{ORBV} which were so far unbeaten in this regard. We call these STBCs ``improved perfect STBCs'' and they are formally defined as follows.
\begin{definition}\label{imp_per}
{\it (Improved perfect STBC)} : An STBC that satisfies C1, C2, C3.1, C3.2 and C4, and has a larger normalized minimum determinant than the existing best comparable perfect STBC is called an improved perfect STBC. 
\end{definition}

\section{Improved perfect STBC for 4 Tx}\label{stbc_4}
 The improved perfect STBC for $4$ transmit antennas, which we call $\mathcal{C}_4$, was first reported in \cite{spcom} but its superior coding gain went unnoticed. $\mathcal{C}_4$ is obtained from the CDA $\mathcal{A} = (\mathbb{Q}(i,\zeta_5)/\mathbb{Q}(i), \tau:\zeta_5 \mapsto \zeta_5^2, i)$ \cite{spcom}, with $\zeta_5$ being the primitive $5^{th}$ root of unity. Its codeword matrix, prior to normalization, has the structure
\begin{equation*}\label{C4}
\textbf{S} = \left[\begin{array}{cccc}
                      a_0 & i\tau(a_{3}) & i \tau^2(a_{2}) & i \tau^{3}(a_{1})\\
                      a_1 & \tau(a_0) & i\tau^2(a_{3}) & i\tau^{3}(a_2)\\
                      a_2 & \tau(a_1) & \tau^2(a_0) & i\tau^{3}(a_3)\\                      
                      a_{3} & \tau(a_{2}) & \tau^2(a_{1}) & \tau^{3}(a_0)\\          
                   \end{array}
\right] 
\end{equation*}
where $a_i =  s_{i1} + s_{i2}\zeta_5 + s_{i2}\zeta_5^2 + s_{i2}\zeta_5^3$, $i=0,1,2,3$, and $s_{ij} \in M$-QAM. Clearly, $\mathcal{C}_4$ satisfies C1. The $\mathbb{Q}(i)$-basis is $\{1,\zeta_5, \zeta_5^2,\zeta_5^3\}$ which is also a $\mathbb{Z}[i]$-basis \cite[p. 158]{paul} for $\mathbb{Z}[i,\zeta_5]$ and $\textbf{R}$, as defined in \eqref{integral_basis}, is   
\begin{equation*}
 \frac{1}{2}\left[\begin{array}{cccc}
                      1 & \zeta_5   & \zeta_5^2 & \zeta_5^3\\
                      1 & \zeta_5^2 & \zeta_5^4 & \zeta_5\\
                      1 & \zeta_5^4 & \zeta_5^3 & \zeta_5^2\\                      
                      1 & \zeta_5^3 & \zeta_5   & \zeta_5^4\\          
                   \end{array}
\right]. 
\end{equation*}
It is clear that C3.1 and C3.2 are satisfied. Noting that $\gamma = i$ has unit modulus, $\mathcal{C}_4$ satisfies C4 as well. It only remains to be seen whether C2 is satisfied. Although this is shown in \cite{spcom}, we provide our version of the proof here for the sake of completeness and the steps of this proof will be used in the next section where the STBC for $6$ transmit antennas is discussed. We first show that $(\mathbb{Q}(i,\zeta_5)/\mathbb{Q}(i), \tau:\zeta_5 \mapsto \zeta_5^2, i)$ is a division algebra and subsequently, application of Proposition \ref{prop1} establishes that the NVD criterion is satisfied.
\begin{proposition}
$\mathcal{A} = (\mathbb{Q}(i,\zeta_5)/\mathbb{Q}(i), \tau:\zeta_5 \mapsto \zeta_5^2, i)$ is a division algebra. 
\end{proposition}
\begin{proof}
 To prove that $\mathcal{A}$ is a CDA, it is sufficient to show that $N_{\mathbb{Q}(i,\zeta_5)/\mathbb{Q}(i)}(a) = \prod_{j=0}^{3}\tau^j(a) \neq i^t$, $t=1,2,3$, $\forall a \in \mathbb{Q}(i,\zeta_5)$. Thus, we have to establish that $i$, $-1$ and $-i$ are not norms in $\mathbb{Q}(i,\zeta_5)/\mathbb{Q}(i)$. Noting that $\zeta_5 + \zeta_5^{-1} = (-1+\sqrt{5})/2$, it is clear that $\mathbb{Q}(i,\sqrt{5}) \subset \mathbb{Q}(i,\zeta_5)$. Since $[\mathbb{Q}(i,\zeta_5) : \mathbb{Q}(i) ] = 4$  and $[\mathbb{Q}(i,\sqrt{5}):\mathbb{Q}(i)]=2$, by the multiplicative formula for tower of fields, $[\mathbb{Q}(i,\zeta_5) :\mathbb{Q}(i,\sqrt{5})]=2$ and $\mathbb{Q}(i,\zeta_5)/\mathbb{Q}(i,\sqrt{5})$ is a Galois extension of degree $2$. Further, since $\zeta_5^4 = \zeta_5^{-1}$, $\tau^2(\zeta_5 + \zeta_5^{-1}) = \zeta_5^{-1} + \zeta_5$ and $\tau^2$ fixes $\mathbb{Q}(i,\sqrt{5})$. So, $Gal(\mathbb{Q}(i,\zeta_5)/\mathbb{Q}(i,\sqrt{5})) = \{1,\tau^2\}$ and $Gal(\mathbb{Q}(i,\sqrt{5})/\mathbb{Q}(i)) = \left\{1,\tau_{\vert\mathbb{Q}(i,\sqrt{5})}\right\}$, where $\tau_{\vert\mathbb{Q}(i,\sqrt{5})}$ denotes ``$\tau$ restricted to $\mathbb{Q}(i,\sqrt{5})$''. If $i$ were a norm in $\mathbb{Q}(i,\zeta_5)/\mathbb{Q}(i)$, then for some $a$ in $\mathbb{Q}(i,\zeta_5)$,
\begin{eqnarray}\nonumber
  i & = &a\tau(a)\tau^2(a)\tau^3(a) \\ \label{norm_i}
    & = & \left(a\tau^2(a)\right)\tau\left(a\tau^2(a) \right).
\end{eqnarray}
But $a\tau^2(a)$ is invariant under $\tau^2$ and hence belongs to $\mathbb{Q}(i,\sqrt{5})$. So, \eqref{norm_i} implies that $i$ is a norm in $\mathbb{Q}(i,\sqrt{5})/\mathbb{Q}(i)$ which is not true \cite{BRV} since $(\mathbb{Q}(i,\sqrt{5})/\mathbb{Q}(i), \tau:\sqrt{5}\mapsto -\sqrt{5}, i)$ is a division algebra. Therefore, $i$ is not a norm in $\mathbb{Q}(i,\zeta_5)$. Likewise, $-i$ is also not a norm in $\mathbb{Q}(i,\sqrt{5})/\mathbb{Q}(i)$ (for if $a\tau(a) =-i$, then $(ia)\tau(ia) = i$ for some $a \in \mathbb{Q}(i,\sqrt{5})$ which is a contradiction) and hence not a norm in $\mathbb{Q}(i,\zeta_5)/\mathbb{Q}(i)$.

Now, it only remains to be seen that $-1$ is not a norm in $\mathbb{Q}(i,\zeta_5)/\mathbb{Q}(i)$. This is proved using class field theory whose usage in proving that a unit is not a norm in the extension field is provided in \cite[Appendix II]{ORBV}. In \cite[Appendix IV]{ORBV}, it is shown that $-1$ is not a norm in $\mathbb{Q}\left(i,2\cos\left(\frac{2\pi}{15}\right)\right)/\mathbb{Q}(i)$. The discriminant (see Appendix I of this paper) of $\mathbb{Q}(i,\zeta_5)/\mathbb{Q}(i)$ is $5^3\mathbb{Z}[i]$. The only prime ideals in $\mathbb{Z}[i]$ that are ramified in $\mathbb{Q}(i,\zeta_5)$ are the ones that divide $125\mathbb{Z}[i]$ and hence divide $5\mathbb{Z}[i]$. These are precisely the prime ideals $(2+i)\mathbb{Z}[i]$ and $(2-i)\mathbb{Z}[i]$. With these facts, the same proof given in \cite[Appendix IV]{ORBV}, with $2$ minor changes, establishes that $-1$ is not a norm in $\mathbb{Q}(i,\zeta_5)/\mathbb{Q}(i)$. The first minor change is that we need to establish that the prime ideal $(-25+12i)\mathbb{Z}[i]$ does not completely split in $\mathbb{Z}[i,\zeta_5]$ whereas in \cite[Appendix IV]{ORBV}, $(-25+12i)\mathbb{Z}[i]$ was required not to be completely split in the ring of integers of $\mathbb{Q}\left(i,2\cos\frac{2\pi}{15}\right)$. That $(-25+12i)\mathbb{Z}[i]$ does not completely split in $\mathbb{Z}[i,\zeta_5]$ is shown in Appendix II. The second change from the proof in \cite[Appendix IV]{ORBV} is that $3\mathbb{Z}[i]$ is not ramified in $\mathbb{Q}(i,\zeta_5)/\mathbb{Q}(i)$ and need not be taken into consideration for evaluating the Hasse norm symbol at ramified places.
\end{proof}

\subsection{Minimum determinant}
 The entries of all the codewords of $\mathcal{C}_4$ prior to normalization of $\textbf{R}$ by $1/2$ belong to $\mathbb{Z}[i,\zeta_5]$, the ring of integers of $\mathbb{Q}(i,\zeta_5)$, and hence the determinant of any codeword difference matrix belongs to $\mathbb{Z}[i,\zeta_5]$. From \eqref{det_f}, the determinant of any codeword difference matrix belongs to $\mathbb{Q}(i)\cap\mathbb{Z}[i,\zeta_5] = \mathbb{Z}[i]$ and so, the minimum determinant is at least $1$. But when the symbols take values from $M$-QAM with an average energy of $E$ units, the nonzero difference between any two symbols is a multiple of $2$. Taking into account a scaling factor of $\frac{1}{4\sqrt{E}}$ so that the expectation of the square of the Euclidean norm of each column of the codeword matrices is unity\footnote{For STBCs like the perfect STBCs, the average energy for transmission of symbols in each time slot is the same and the energy constraint \eqref{energy_con} translates to the requirement that the expectation of the square of the Euclidean norm of each column of codeword matrices be unity.} (see Definition \ref{min_det_stbc}), the normalized minimum determinant of $\mathcal{C}_4$ is $\left(\frac{2}{4\sqrt{E}}\right)^8 = \frac{1}{256E^4}$ which is significantly larger than the normalized minimum determinant of the perfect STBC for $4$ transmit antennas that stands at $\frac{1}{1125E^4}$ \cite{ORBV}. A result of this larger minimum determinant is a superior error performance compared to the perfect STBC and this is evident in Fig. \ref{fig1} which gives a comparison of the error performance of the two STBCs for $4$-QAM.

\section{$\mathcal{C}_6$ - Improved perfect STBC for 6 Tx} \label{stbc_6}
 $\mathcal{C}_6$ is obtained from the algebra $\mathcal{A} = (\mathbb{Q}(\omega,\zeta_7)/\mathbb{Q}(\omega), \tau:\zeta_7 \mapsto \zeta_7^3, -\omega)$ with $\zeta_7$ being the primitive $7^{th}$ root of unity. Its codeword matrix has the structure shown in \eqref{c6} at the top of the page with $a_i =  s_{i1} + s_{i2}\zeta_7 + s_{i3}\zeta_7^2 + s_{i4}\zeta_7^3+s_{i5}\zeta_7^4+s_{i6}\zeta_7^5$, $i=0,1,2,\cdots,5$, and $s_{ij} \in M$-HEX. Clearly, $\mathcal{C}_6$ is full-rate since $\{1,\zeta_7, \zeta_7^2,\zeta_7^3,\zeta_7^4,\zeta_7^5\}$ is a $\mathbb{Z}[\omega]$-basis for $\mathbb{Z}[\omega,\zeta_7]$. $\textbf{R}$ (as defined in \eqref{integral_basis}) is   
\begin{equation*}
 \frac{1}{\sqrt{6}}\left[\begin{array}{cccccc}
                      1 & \zeta_7     & \zeta_7^2  & \zeta_7^3 & \zeta_7^4 & \zeta_7^5\\
                      1 & \zeta_7^3   & \zeta_7^6  & \zeta_7^2 & \zeta_7^5 & \zeta_7  \\
                      1 & \zeta_7^2   & \zeta_7^4  & \zeta_7^6 & \zeta_7   & \zeta_7^3\\               
                      1 & \zeta_7^6   & \zeta_7^5  & \zeta_7^4 & \zeta_7^3 & \zeta_7^2\\ 
                      1 & \zeta_7^4   & \zeta_7    & \zeta_7^5 & \zeta_7^2 & \zeta_7^6\\  
                      1 & \zeta_7^5   & \zeta_7^3  & \zeta_7   & \zeta_7^6 & \zeta_7^4\\      
                   \end{array}
\right] 
\end{equation*}
and it is clear that the norm of each row and column of $\textbf{R}$ is equal to $1$. Noting that $\gamma = -\omega$ has unit modulus, $\mathcal{C}_6$ satisfies C3.1, C3.2 and C4. To show that the NVD criterion is also satisfied, it is sufficient to show that $(\mathbb{Q}(\omega,\zeta_7)/\mathbb{Q}(\omega), \tau:\zeta_7 \mapsto \zeta_7^3, -\omega)$ is a division algebra following which the application of Proposition \ref{prop1} establishes that the NVD criterion is satisfied.

\begin{figure}
\centering
\includegraphics[width=3in,height=2.5in]{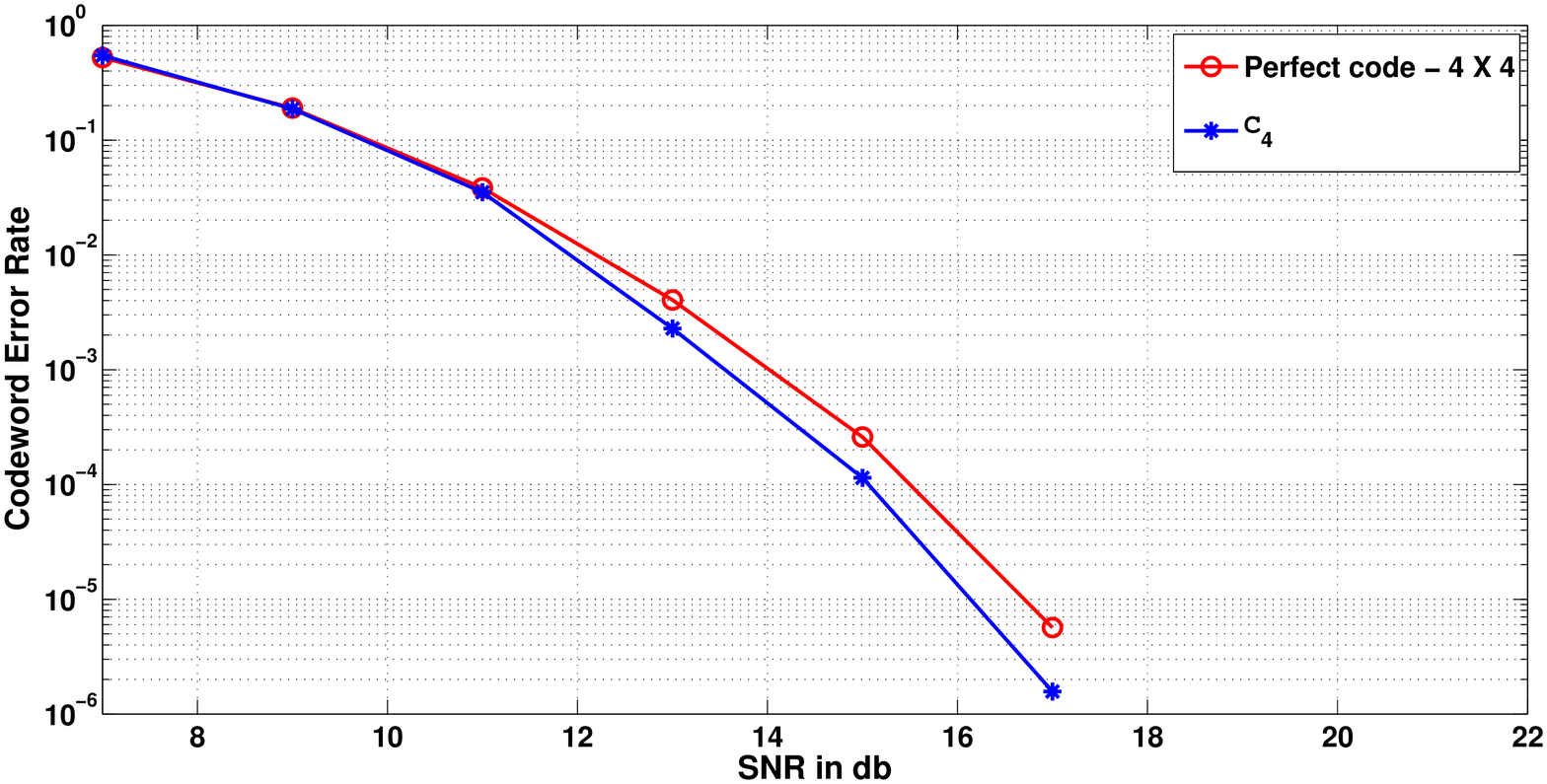}
\caption{CER performance of the Perfect STBC and $\mathcal{C}_4$ for the $4 \times 4$ system with $4$-QAM }
\label{fig1}

\end{figure}

\begin{table*}
\begin{center}
\begin{threeparttable}
\begin{tabular}{|c|c|c|c|c|} \cline{1-5}
\multirow{2}{*}{ $\#$ Tx antennas} & \multirow{2}{*}{STBC $\mathcal{S}$}  & Constellation &\multirow{2}{*}{$\delta_{min}(\mathcal{S})$} & \multirow{2}{*}{Approximately Universal?}  \\ 
& & (average energy $E$) & & \\ \hline \hline
 \multirow{4}{*}{$4$} & \multirow{2}{*}{Perfect Code \cite{ORBV}}  & \multirow{2}{*}{QAM} &   	\multirow{2}{*}{$\frac{1}{1125E^4}$}   &  \multirow{2}{*}{Yes}\\ 
& & & & \\ \cline{2-5}
 & \multirow{2}{*}{$\mathcal{C}_4$ \cite{spcom}} & \multirow{2}{*}{QAM} & \multirow{2}{*}{$ \frac{1}{256E^4}$} & \multirow{2}{*}{Yes} \\ 
& & & & \\ \hline \hline
 \multirow{4}{*}{$6$} & \multirow{2}{*}{Perfect STBC \cite{ORBV}} & \multirow{2}{*}{HEX} &  \multirow{2}{*}{$\frac{1}{3^67^5E^6} \leq \delta_{min}\leq \frac{1}{3^67^4E^6}$} & \multirow{2}{*}{Yes} \\ 
& & & & \\ \cline{2-5}
 & \multirow{2}{*}{$ \mathcal{C}_6 $ } & \multirow{2}{*}{HEX} & \multirow{2}{*}{$\frac{1}{3^{12}E^6}$} & \multirow{2}{*}{Yes} \\
 &  &  & & \\  \hline
\end{tabular}
\end{threeparttable}
\end{center}
\caption{Comparison between the improved perfect STBCs and the perfect STBCs.}
\label{tab1}
\hrule
\end{table*}

\begin{proposition}
$\mathcal{A} = (\mathbb{Q}(\omega,\zeta_7)/\mathbb{Q}(\omega), \tau:\zeta_7 \mapsto \zeta_7^3, -\omega)$ is a division algebra. 
\end{proposition}
\begin{proof}
 To prove that $\mathcal{A}$ is a CDA, it is sufficient to show that $N_{\mathbb{Q}(\omega,\zeta_7)/\mathbb{Q}(\omega)}(a) = \prod_{j=0}^{5}\tau^j(a) \neq (-\omega)^t$, $t=1,2,\cdots,5$, $\forall a \in \mathbb{Q}(\omega,\zeta_7)$. Hence, it is to be established that $\pm \omega$, $\pm \omega^2$, $-1$ are not norms in $\mathbb{Q}(\omega,\zeta_7)/\mathbb{Q}(\omega)$. We note that $\mathbb{Q}(\omega,\zeta_7+\zeta_7^{-1}) \subset \mathbb{Q}(\omega,\zeta_7)$. Since $[\mathbb{Q}(\omega,\zeta_7) : \mathbb{Q}(\omega) ] = 6$  and $[\mathbb{Q}(\omega,\zeta_7+\zeta_7^{-1}):\mathbb{Q}(\omega)]=3$, by the multiplicative formula for tower of fields, $[\mathbb{Q}(\omega,\zeta_7) :\mathbb{Q}(\omega,\zeta_7+\zeta_7^{-1})]=2$ and $\mathbb{Q}(\omega,\zeta_7)/\mathbb{Q}(\omega,\zeta_7+\zeta_7^{-1})$ is a Galois extension of degree $2$. Further, $\tau^3(\zeta_7 + \zeta_7^{-1}) = \zeta_7^{-1} + \zeta_7$ (since $\zeta_7^{-1} = \zeta_7^6$) and $\tau^3$ fixes $\mathbb{Q}(\omega,\zeta_7+\zeta_7^{-1})$. So, $Gal(\mathbb{Q}(\omega,\zeta_7)/\mathbb{Q}(\omega,\zeta_7+\zeta_7^{-1})) = \{1,\tau^3\}$ and $Gal(\mathbb{Q}(\omega,\zeta_7+\zeta_7^{-1})/\mathbb{Q}(\omega)) = \left\{1,\tau_{\vert\mathbb{Q}(\omega,\zeta_7+\zeta_7^{-1})},\tau^2_{\vert\mathbb{Q}(\omega,\zeta_7+\zeta_7^{-1})}\right\}$. So, if $\pm\omega$ were a norm in $\mathbb{Q}(\omega,\zeta_7)/\mathbb{Q}(\omega)$, then for some $a$ in $\mathbb{Q}(\omega,\zeta_7)$,
\begin{eqnarray}\nonumber
 \pm \omega & = &a\tau(a)\tau^2(a)\tau^3(a)\tau^4(a)\tau^5(a) \\ \label{norm_omega}
    & = & \left(a\tau^3(a)\right)\tau\left(a\tau^3(a) \right)\tau^2\left(a\tau^3(a) \right).
\end{eqnarray}
But $a\tau^3(a)$ is invariant under $\tau^3$ and hence belongs to $\mathbb{Q}(\omega,\zeta_7+\zeta_7^{-1})$. So, \eqref{norm_omega} implies that $\omega$ is a norm in $\mathbb{Q}(\omega,\zeta_7+\zeta_7^{-1})/\mathbb{Q}(\omega)$ which is not true \cite{BRV} since $(\mathbb{Q}(\omega,\zeta_7+\zeta_7^{-1})/\mathbb{Q}(\omega), \tau^2:\zeta_7+\zeta_7^{-1} \mapsto \zeta_7^2+\zeta_7^{-2}, \omega)$ is a division algebra ($-\omega$ not being a norm naturally follows). Therefore, $\pm \omega$ is not a norm in $\mathbb{Q}(\omega,\zeta_7)$. Likewise, $\pm \omega^2$ is also not a norm in $\mathbb{Q}(\omega,\zeta_7+\zeta_7^{-1})/\mathbb{Q}(\omega)$ and hence not a norm in $\mathbb{Q}(\omega,\zeta_7)/\mathbb{Q}(\omega)$.

Now, it only remains to be seen that $-1$ is not a norm in $\mathbb{Q}(\omega,\zeta_7)/\mathbb{Q}(\omega)$. This is again proved using class field theory. In \cite[Appendix V]{ORBV}, it is shown that $-1$ is not a norm in $\mathbb{Q}\left((\omega,2\cos\left(\frac{2\pi}{28}\right)\right)/\mathbb{Q}(\omega))$. The discriminant of $\mathbb{Q}(\omega,\zeta_7)/\mathbb{Q}(\omega)$ is $ 7^5\mathbb{Z}[\omega]$. The only prime ideals in $\mathbb{Z}(\omega)$ that are ramified in $\mathbb{Q}(\omega,\zeta_7)$ are the ones that divide $7^5\mathbb{Z}[\omega]$ and hence divide $7\mathbb{Z}[\omega]$. These are precisely the prime ideals $(3+\omega)\mathbb{Z}[\omega]$ and $(2-\omega)\mathbb{Z}[\omega]$. Using these facts, the same proof given in \cite[Appendix V]{ORBV}, with $2$ minor changes, establishes that $-1$ is not a norm in $\mathbb{Q}(\omega,\zeta_7)/\mathbb{Q}(\omega)$. The first change is that we are required to show that the prime ideal $(3-8\omega)\mathbb{Z}[\omega]$ is not completely split in $\mathbb{Z}[\omega,\zeta_7]$ whereas in \cite[Appendix V]{ORBV}, $(3-8\omega)\mathbb{Z}[\omega]$ was required to be not completely split in the ring of integers of $\mathbb{Q}\left(\omega,2\cos\frac{2\pi}{28}\right)$. It is shown in Appendix III of this paper that $(3-8\omega)\mathbb{Z}[\omega]$ is not completely split in $\mathbb{Z}[\omega,\zeta_7]$. The second change from the proof in \cite[Appendix V]{ORBV} is that $2\mathbb{Z}[\omega]$ is not ramified in $\mathbb{Q}(\omega,\zeta_7)/\mathbb{Q}(\omega)$ and need not be taken into consideration for evaluating the Hasse norm symbol at ramified places.
\end{proof}

\subsection{Minimum Determinant}

The entries of all the codewords of $\mathcal{C}_6$, prior to normalization of $\textbf{R}$ by $1/\sqrt{6}$, belong to $\mathbb{Z}[\omega,\zeta_7]$, the ring of integers of $\mathbb{Q}(\omega,\zeta_7)$, and hence the determinant of any codeword difference matrix belongs to $\mathbb{Q}(\omega)\cap\mathbb{Z}[\omega,\zeta_7] = \mathbb{Z}[\omega]$. So, the minimum determinant is guaranteed to be at least $1$. But since the symbols take values from $M$-HEX with an average energy of $E$ units, the nonzero difference between any two symbols is a multiple of $2$. Taking into account a normalizing factor of $\frac{1}{6\sqrt{E}}$, the normalized minimum determinant of $\mathcal{C}_6$ is $\left(\frac{2}{6\sqrt{E}}\right)^{12} = \frac{1}{3^{12}E^6}$ which is significantly larger than the normalized minimum determinant of the perfect STBC for $6$ transmit antennas that is upper bounded by $\frac{1}{3^67^4E^6}$ \cite{ORBV}. The normalized minimum determinants of the improved perfect STBCs and the perfect STBCs are tabulated in Table \ref{tab1}.

{\it Remarks}: We have restricted our construction of the improved perfect STBCs to just $4$ and $6$ transmit antennas. The usage of cyclotomic extensions of $\mathbb{Q}(i)$ and $\mathbb{Q}(\omega)$ was the reason we were able to obtain STBCs with larger normalized minimum determinants than that of perfect STBCs for $4$ and $6$ transmit antennas. However, for $n_t=2,3$, one cannot obtain CDAs of degree $n_t$ over $\mathbb{Q}(i)$ or $\mathbb{Q}(\omega)$ using cyclotomic extensions (with $\zeta_3 =\omega$, $(\mathbb{Q}(i,\omega)/\mathbb{Q}(i),\tau:\omega \to \omega^2, i)$ is not a division algebra). So, for $2$ and $3$ transmit antennas, the existing perfect STBCs \cite{ORBV} remain the best with respect to coding gain. For other values of $n_t$, $\gamma$ cannot be a unit in $\mathbb{Z}[i]$ or $\mathbb{Z}[\omega]$ for the algebra to be a division algebra. However, the approach taken in \cite{new_per}, where $\gamma$ is not restricted to be in $\mathbb{Z}[i]$ or $\mathbb{Z}[\omega]$, can still be taken to investigate if new STBCs with larger coding gains can be obtained for arbitrary number of transmit antennas.

\section{Concluding Remarks}
In this paper, we presented a modified shaping criterion in the design of STBCs that enabled us to propose two STBCs, one each for $4$ and $6$ transmit antennas, that have the best known normalized minimum determinants in their comparable class. This shaping criterion can be employed to see if better STBCs, in terms of coding gain, can be obtained for arbitrary number of transmit antennas. 

\section*{Appendix I}
\begin{center}
\textsc{Number Theory basics and definitions}
\end{center}

We consider a number field $\mathbb{F}$ that is a finite extension of $\mathbb{Q}$. Its ring of integers $\mathcal{O}_{\mathbb{F}}$ is given by $\mathcal{O}_{\mathbb{F}} =  \{a \in \mathbb{F} ~\vert~ f(a) = 0, f \in \mathbb{Z}_{monic}[X]\}$ where $\mathbb{Z}_{monic}[X]$ is the set of {\it monic} polynomials in the variable $X$ with coefficients in $\mathbb{Z}$. Let the Galois extension of $\mathbb{F}$ of degree $n$ be denoted by $\mathbb{K}$ whose ring of integers is denoted by $\mathcal{O}_\mathbb{K}$ and $Gal(\mathbb{K}/\mathbb{F}) = \{\sigma_1,\sigma_2,\cdots,\sigma_n\}$. It is well-known that for any $a$ in $\mathbb{K}$, if $\sigma_i(a) = a$, $\forall i=1,\cdots,n$, then $a \in \mathbb{F}$. Let $\{\theta_1,\theta_2,\cdots,\theta_n \}$ be the $\mathcal{O}_{\mathbb{F}}$-basis of $\mathcal{O}_\mathbb{K}$.

 {\it Trace of an element}: The trace of an element $a$ in $\mathbb{K}/\mathbb{F}$, denoted by $T_{\mathbb{K}/\mathbb{F}}(a)$, is $\sum_{i=1}^{n}\sigma_i(a)$ and belongs to $\mathbb{F}$.

 {\it Norm of an element}: The norm of an element $a $ in $\mathbb{K}/\mathbb{F}$, denoted by $N_{\mathbb{K}/\mathbb{F}}(a)$, is $\prod_{i=1}^{n}\sigma_i(a)$ and belongs to $\mathbb{F}$.

 {\it Discriminant of a basis} \cite[p. 25]{paul}: For a chosen $\mathbb{F}$-basis $\{b_1,b_2,\cdots,b_n\}$, its discriminant, denoted by $\Delta(b_1,b_2,\cdots,b_n)$, is the determinant of the $n\times n$ matrix $\textbf{M}$ whose $(i,j)^{th}$ entry is $T_{\mathbb{K}/\mathbb{F}}(b_ib_j)$. 

 {\it Discriminant of $\mathbb{K}/\mathbb{F}$} \cite[p. 148]{paul}: The discriminant of $\mathbb{K}/\mathbb{F}$ is the ideal $\Delta(\theta_1,\theta_2,\cdots,\theta_n)\mathcal{O}_\mathbb{F}$.

 {\it Prime ideal}: An ideal $\mathfrak{p}$ of a ring $\mathcal{R}$ is a prime ideal if it has the following properties.
\begin{itemize}
 \item If $a,b \in \mathcal{R}$ such that $ab \in \mathfrak{p}$, then either $a \in \mathfrak{p}$ or $b \in \mathfrak{p}$.
\item $\mathfrak{p}$ is not $\mathcal{R}$ itself.
 \end{itemize}
A nonzero principal ideal is prime if and only if it is generated by a prime element.

{\it Prime elements of $\mathbb{Z}[i]$}: A Gaussian integer $a+ib$, $a,b \in \mathbb{Z}$ is a Gaussian prime if and only if either
\begin{itemize}
 \item one of $a$, $b$ is zero and the other is a prime number of the form $\pm(4n+3)$, with $n$ a nonnegative integer, or 
\item both $a$ and $b$ are nonzero and $a^2+b^2$ is a prime number (which will not be of the form $4n+3$).
\end{itemize}

{\it Prime elements of $\mathbb{Z}[\omega]$}: An Eisenstein integer $z = a+\omega b$, $a,b \in \mathbb{Z}$ is an Eisenstein prime if and only if either
\begin{itemize}
 \item one of $a$, $b$ is zero and $z$ is equal to the product of a unit and a natural prime of the form $3n-1$, or
\item both $a$ and $b$ are nonzero and $\vert z \vert ^2 = a^2 -ab + b^2$ is a natural prime (which is necessarily congruent to $0$ or $1$ mod $3$).
\end{itemize} 

{\it Relative prime ideals}: Ideals $A$ and $B$ of a ring $\mathcal{R}$ are said to be relatively prime (coprime) if $A+B = \mathcal{R}$. It follows that coprime ideals $A$ and $B$ of $\mathcal{R}$ satisfy $AB = A\cap B$. 

{\it Dedekind domain}: An integral domain $\mathcal{R}$ which is not a field is called a Dedekind domain if every nonzero proper ideal factors into prime ideals. The ring of integers of a number field is a Dedekind domain.

 {\it Ideal factorization in extensions} \cite[p. 144]{paul}: Let $\mathfrak{p}$ be a nonzero prime ideal in $\mathbb{O}_\mathbb{F}$. Then, in the extension field $\mathbb{K}$ (not necessarily a Galois extension),  
\begin{equation*}
 \mathfrak{p}\mathcal{O}_\mathbb{K} = \prod_{i=1}^{g}\mathfrak{B}_i^{e(\mathfrak{B}_i/\mathfrak{p})}
\end{equation*}
where $\mathfrak{B}_i \subset \mathcal{O}_{\mathbb{K}}$ are prime ideals (finite in number) in $\mathcal{O}_{\mathbb{K}}$, $e(\mathfrak{B}_i/\mathfrak{p})$ is a non-negative integer called the {\it ramification index} of $\mathfrak{B}_i$ over $\mathfrak{p}$ and is the exact power of $\mathfrak{B}_i$ that divides $\mathfrak{p}\mathcal{O}_\mathbb{K}$. $\mathfrak{B}_i$ is {\it said to lie above} $\mathfrak{p}$ in $\mathcal{O}_\mathbb{K}$. This factorization is {\it unique} up to order of the factors since $\mathcal{O}_\mathbb{K}$ is a Dedekind domain. 

{\it Inertia degree or residue class degree} \cite[p. 105]{paul}: Let $\mathfrak{p}$ be a prime ideal in $\mathcal{O}_\mathbb{F}$ that factors into prime ideals in $\mathcal{O}_\mathbb{K} $ as $\mathfrak{p}\mathbb{O}_\mathbb{K} = \prod_{i=1}^{g}\mathfrak{B}_i^{e(\mathfrak{B}_i/\mathfrak{p})}$. Then, the inertia degree $f(\mathfrak{B}_i/\mathfrak{p})$ of $\mathfrak{B}_i$ over $\mathfrak{p}$ is a non-negative integer given by 
\begin{equation*}
 f(\mathfrak{B}_i/\mathfrak{p}) = [\mathcal{O}_\mathbb{K}/\mathfrak{B}_i:\mathcal{O}_\mathbb{F}/\mathfrak{p}].
\end{equation*}
It follows that \cite[p. 144]{paul}
\begin{equation*}
 [\mathbb{K}:\mathbb{F}] = \sum_{i=1}^{g}e(\mathfrak{B}_i/\mathfrak{p})f(\mathfrak{B}_i/\mathfrak{p}).
\end{equation*}

\begin{corollary}\label{tower_mul}
 \cite[p. 191]{paulo} Consider a tower of field extensions $\mathbb{F} \subset \mathbb{K} \subset \mathbb{L}$ with the ring of integers $\mathcal{O}_\mathbb{F} \subset \mathcal{O}_\mathbb{K} \subset \mathcal{O}_\mathbb{L}$. Let $\mathfrak{p}$ be a prime ideal of $\mathcal{O}_\mathbb{F}$, $\mathfrak{B}_\mathbb{K}$ a prime ideal of $\mathcal{O}_\mathbb{K}$ lying above $\mathfrak{p}$ and $\mathfrak{B}_\mathbb{L}$ a prime ideal of $\mathcal{O}_\mathbb{L}$ lying above $\mathfrak{B}_\mathbb{K}$. Then, the ramification index and inertia degree are multiplicative in the tower, i.e.,
\begin{eqnarray*}
 e(\mathfrak{B}_\mathbb{L}/\mathfrak{p}) & = &e(\mathfrak{B}_\mathbb{L}/\mathfrak{B}_\mathbb{K})e(\mathfrak{B}_\mathbb{K}/\mathfrak{p})\\
 f(\mathfrak{B}_\mathbb{L}/\mathfrak{p}) & = &f(\mathfrak{B}_\mathbb{L}/\mathfrak{B}_\mathbb{K})f(\mathfrak{B}_\mathbb{K}/\mathfrak{p}).
\end{eqnarray*}
\end{corollary}

For Galois extensions $\mathbb{K}/\mathbb{F}$, $e(\mathfrak{B}_1/\mathfrak{p}) = e(\mathfrak{B}_2/\mathfrak{p}) = \cdots = e(\mathfrak{B}_g/\mathfrak{p})$ and $f(\mathfrak{B}_1/\mathfrak{p}) = f(\mathfrak{B}_2/\mathfrak{p}) = \cdots = f(\mathfrak{B}_g/\mathfrak{p})$ \cite[p. 152]{paul}. In such a case, we simply denote the ramification index and the inertia degree by $e$ and $f$, respectively, and 
\begin{equation}\label{prime_fac}
 [\mathbb{K}:\mathbb{F}] = n = efg.
\end{equation}

{\it Definition}: Let $\mathfrak{p}$ be a prime ideal in $\mathcal{O}_\mathbb{F}$ that factors into prime ideals of $\mathcal{O}_\mathbb{K}$ in the Galois extension field $\mathbb{K}$ as $\mathfrak{p}\mathbb{O}_\mathbb{K} = \prod_{i=1}^{g}\mathfrak{B}_i^e$ with an inertia degree $f$. Then,
\begin{itemize}
 \item $\mathfrak{p}$ is {\it ramified} in $\mathbb{K}$ if $e >1$.
\item $\mathfrak{p}$ is {\it totally ramified} in $\mathbb{K}$ if $e =n$, $g=1$, $f=1$.
\item $\mathfrak{p}$ {\it splits} in $\mathcal{O}_\mathbb{K}$ if $g>1$.
\item $\mathfrak{p}$ {\it splits completely} in $\mathcal{O}_\mathbb{K}$ if $e =1$, $g=n$, $f=1$.
\item $\mathfrak{p}$ is {\it inert} in $\mathcal{O}_\mathbb{K}$ if $e =1$, $g=1$.
\end{itemize}

{\it Corollary} \cite[P. 148]{paul}: A prime ideal $\mathfrak{p}$ of $\mathcal{O}_\mathbb{F}$ is ramified in $\mathbb{K}$ if and only if it divides the discriminant of $\mathbb{K}/\mathbb{F}$. 

Let $\theta \in \mathcal{O}_\mathbb{K}$ such that $\mathbb{K} = \mathbb{F}(\theta)$ (not necessarily a Galois extension) with the minimal polynomial of $\theta$ being $p(X) \in \mathcal{O}_\mathbb{F}[X]$. The {\it conductor} of the ring $\mathcal{O}_\mathbb{F}[\theta]$ is the largest ideal $\mathfrak{F}$ of $\mathcal{O}_\mathbb{K}$ that is contained in $\mathcal{O}_\mathbb{F}[\theta]$.
 
 \begin{proposition} \label{prop4}\cite[p. 47]{neukirch} Let $p$ be a prime integer of $\mathcal{O}_\mathbb{F}$ such that $\mathfrak{p} = p\mathcal{O}_\mathbb{F}$ is a prime ideal of $\mathcal{O}_\mathbb{F}$ and $\mathfrak{p}\mathcal{O}_\mathbb{K}$ is relatively prime to the conductor of $\mathcal{O}_\mathbb{F}[\theta]$, and let $\bar{p}(X) = \bar{p}_1(X)^{e_1}\bar{p}_2(X)^{e_2}\cdots\bar{p}_g(X)^{e_g}$ be the factorization of the polynomial $\bar{p}(X) = p(X) \textrm{ mod } p$ into monic irreducibles $\bar{p}_i(X) = p_i(X) \textrm{ mod }p$ over the residue class field $\mathcal{O}_\mathbb{F}/\mathfrak{p}$, with all the $p_i(X) \in \mathcal{O}_\mathbb{F}[X]$ and monic. Then, $\mathfrak{B}_i=\mathfrak{p}\mathcal{O}_\mathbb{K} + p_i(\theta)\mathcal{O}_\mathbb{K}$, $i=1,...,g$, are the different prime ideals of $\mathcal{O}_\mathbb{K}$ above $\mathfrak{p}$. The inertia degree $f(\mathfrak{B}_i/\mathfrak{p})$ of $\mathfrak{B}_i$ over $\mathfrak{p}$ is the degree of $\bar{p}_i(X)$, and one has
 \begin{equation*}
  \mathfrak{p}\mathcal{O}_\mathbb{K} = \mathfrak{B}_1^{e_1}\mathfrak{B}_2^{e_2}\cdots\mathfrak{B}_g^{e_g}.
 \end{equation*}
\end{proposition}

\begin{theorem} \label{thm1}
\cite[Theorem 2.47]{fin_field} Let $\mathbb{F}_q$ be a finite field with $q$ elements and characteristic $p$, $n$ a natural number such that $p$ does not divide $n$. The $n^{th}$ cyclotomic polynomial $\Phi_n(X)$ factorizes over $\mathbb{F}_q$ as a product of irreducible factors all of the same degree $d$ where $d$ is the order of $q$ mod $n$ ($d$ is the smallest positive integer such that $ q^d \equiv 1 \textrm{ mod } n$). 
 \end{theorem}

\section*{Appendix II}
\begin{center}
\textsc{Proof that $(-25 + 12i)\mathbb{Z}[i]$ does not split completely in $\mathbb{Z}[i,\zeta_5]$}
\end{center}

Let $\mathfrak{p}_{769} = (-25 + 12i)\mathbb{Z}[i]$ which is a prime ideal of $\mathbb{Z}[i]$. The discriminant of $\mathbb{Q}(i,\zeta_5)/\mathbb{Q}(i)$ is $125\mathbb{Z}[i]$ and clearly $\mathfrak{p}_{769}$ does not divide $125\mathbb{Z}[i]$. So, $\mathfrak{p}_{769}$ is not ramified in $\mathbb{Q}(i,\zeta_5)$. We have the following tower of Galois field extensions.
\begin{eqnarray*}
 \mathbb{Q} & \subset & \mathbb{Q}(i) \subset  \mathbb{Q}(i,\zeta_5),\\
 \mathbb{Q} & \subset & \mathbb{Q}(\zeta_5) \subset  \mathbb{Q}(i,\zeta_5)
\end{eqnarray*}
where $[\mathbb{Q}(i,\zeta_5) : \mathbb{Q}] = 8$, $[\mathbb{Q}(\zeta_5) : \mathbb{Q}] = 4$. The prime ideal $769\mathbb{Z}$ splits into two prime ideals $\mathfrak{p}_{769} = (-25 + 12i)\mathbb{Z}[i]$ and $\mathfrak{q}_{769}=(-25 - 12i)\mathbb{Z}[i]$ in $\mathbb{Z}[i]$. From Corollary \ref{tower_mul} and \eqref{prime_fac} in Appendix I, $769\mathbb{Z}$ splits completely in $\mathbb{Z}[i,\zeta_5]$ if and only if $\mathfrak{p}_{769}$ and $\mathfrak{q}_{769}$ split completely in $\mathbb{Z}[i,\zeta_5]$. Also, $769\mathbb{Z}$ splits completely in $\mathbb{Z}[i,\zeta_5]$ if and only if it splits completely in $\mathbb{Z}[\zeta_5]$.

So, it is sufficient to prove that the ideal $769\mathbb{Z}$ does not split completely in $\mathbb{Z}[\zeta_5]$. For this purpose, we consider the minimal polynomial of $\zeta_5$ over $\mathbb{Q}$ which is $X^4+X^3+X^2+X+1$ and is also the $5^{th}$ cyclotomic polynomial $\Phi_5(X)$. From Theorem \ref{thm1} in Appendix I, $\Phi_5(X)$ splits into only 2 irreducible monic factors over $\mathbb{F}_{769}$, each of degree $2$. Hence, from Proposition \ref{prop4}, it is clear that $769\mathbb{Z}$ does not split completely in $\mathbb{Z}[\zeta_5]$. This establishes that $(-25 + 12i)\mathbb{Z}[i]$ does not split completely in $\mathbb{Z}[i,\zeta_5]$. 

\section*{Appendix III}
\begin{center}
\textsc{Proof that $(3 -8\omega)\mathbb{Z}[\omega]$ does not split completely in $\mathbb{Z}[\omega,\zeta_7]$}
\end{center}

Let $\mathfrak{p}_{97} = (3 -8\omega)\mathbb{Z}[\omega]$ which is a prime ideal of $\mathbb{Z}[\omega]$. The discriminant of $\mathbb{Q}(\omega,\zeta_5)/\mathbb{Q}(\omega)$ is $7^5\mathbb{Z}[\omega]$ and clearly $\mathfrak{p}_{97}$ does not divide $7^5\mathbb{Z}[\omega]$. So, $\mathfrak{p}_{97}$ is not ramified in $\mathbb{Q}(\omega,\zeta_7)$. We have the following Galois field extensions.
\begin{eqnarray*}
 \mathbb{Q} & \subset & \mathbb{Q}(\omega) \subset \mathbb{Q}(\omega,\zeta_7),\\
 \mathbb{Q} & \subset & \mathbb{Q}(\zeta_7) \subset \mathbb{Q}(\omega,\zeta_7)
\end{eqnarray*}
where $[\mathbb{Q}(\omega,\zeta_7) : \mathbb{Q}] = 12$, $[\mathbb{Q}(\zeta_7) : \mathbb{Q}] = 6$. The prime ideal $97\mathbb{Z}$ splits into two prime ideals $\mathfrak{p}_{97} = (3 -8\omega)\mathbb{Z}[\omega]$ and $\mathfrak{q}_{97}=(3 -8\omega^2)\mathbb{Z}[\omega]$ in $\mathbb{Z}[\omega]$. It is clear from the Corollary \ref{tower_mul} and \eqref{prime_fac} in Appendix I that $97\mathbb{Z}$ splits completely in $\mathbb{Z}[\omega,\zeta_7]$ if and only if $\mathfrak{p}_{97}$ and $\mathfrak{q}_{97}$ split completely in $\mathbb{Z}[\omega,\zeta_7]$. Also, $97\mathbb{Z}$ splits completely in $\mathbb{Z}[\omega,\zeta_7]$ if and only if it splits completely in $\mathbb{Z}[\zeta_7]$.

So, it suffices to prove that the ideal $97\mathbb{Z}$ does not split completely in $\mathbb{Z}[\zeta_7]$. For this purpose, we consider the minimal polynomial of $\zeta_7$ over $\mathbb{Q}$ which is $X^6+X^5+X^4+X^3+X^2+X+1$ and is also the $7^{th}$ cyclotomic polynomial $\Phi_7(X)$. From Theorem \ref{thm1} in Appendix I, $\Phi_7(X)$ splits into only 3 irreducible monic factors, each of degree $2$ over $\mathbb{F}_{97}$. Hence, from Proposition \ref{prop4}, it is clear that $97\mathbb{Z}$ does not split completely in $\mathbb{Z}[\zeta_7]$. This establishes that $(3 -8\omega)\mathbb{Z}[\omega]$ does not split completely in $\mathbb{Z}[\omega,\zeta_7]$.


\begin{thebibliography}{1}

\bibitem{ORBV}
F. Oggier, G. Rekaya, J. C. Belfiore, and E. Viterbo, ``Perfect space time block codes,'' \emph{IEEE Trans. Inf. Theory,} vol. 52, no. 9, pp. 3885-3902, Sep. 2006.

\bibitem{elia} 
P. Elia, K. R. Kumar, S. A. Pawar, P. V. Kumar, and H.-F. Lu, ``Explicit Space-Time Codes Achieving the Diversity-Multiplexing Gain Tradeoff,'' \emph{IEEE Trans. Inf. Theory}, vol. 52, no. 9, pp. 3869-3884, Sep. 2006.

\bibitem{new_per}
P. Elia, B. A. Sethuraman, and P. V. Kumar, ``Perfect Space-Time Codes for Any Number of Antennas,'' \emph{IEEE Trans. Inf. Theory}, vol. 53, no. 11, pp. 3853-3868, Nov. 2007.


\bibitem{spcom}
F. Oggier, C. Hollanti, and R. Vehkalahti, ``An algebraic MIDO-MISO code
  construction,'' in Proc. \emph{Int. Conf.  Signal
  Process. and Commun. (SPCOM 2010)}, Bangalore, India, July 2010.

\bibitem{HaH}
B. Hassibi and B. Hochwald, ``High-rate codes that are
linear in space and time,'' {\it IEEE Trans. Inf. Theory}, vol. 48, no. 7, pp. 1804-1824, July 2002.

\bibitem{TSC}
V. Tarokh, N. Seshadri, and A. R. Calderbank, ``Space time codes for high date rate wireless communication : performance criterion and code construction,''
\emph{IEEE Trans. Inf. Theory.}, vol.\ 44, no. 2, pp.\ 744 - 765, Mar. 1998.

\bibitem{tse}
 L. Zheng and D. Tse, ``Diversity and Multiplexing: A Fundamental Tradeoff in Multiple-Antenna Channels,'' \emph{IEEE Trans. Inf. Theory}, vol. 49, no. 5, pp. 1073-1096, May 2003. 

\bibitem{BRV}
J. C. Belfiore, G. Rekaya, and E. Viterbo, ``The Golden Code: A $2\times2$ full rate space-time code with non-vanishing determinants,'' \emph{IEEE Trans. Inf. Theory}, vol. 51, no. 4, pp. 1432-1436, Apr. 2005.

\bibitem{damen}
M. O. Damen, A. Tewfik, and J.-C. Belfiore, ``A construction of a space-time code based on number theory,'' \emph{IEEE Trans. Inf. Theory}, vol. 48, no. 3, pp. 753-761, Mar. 2002.

\bibitem{jacobson}
N. Jacobson, {\it Basic Algebra II}. 2nd ed. New York: W.H. Freeman, 1985.

\bibitem{sethuraman}
B. A. Sethuraman, B. S. Rajan, and V. Shashidhar, ``Full-diversity, High-rate Space-Time Block Codes from Division Algebras,'' \emph{IEEE Trans. Inf. Theory}, vol. 49, no. 10, pp.2596-2616, Oct. 2003.


\bibitem{max_order}
R. Vehkalahti, C. Hollanti, J. Lahtonen, and K. Ranto, ``On the Densest MIMO Lattices From Cyclic Division Algebras,'' \emph{IEEE Trans. Inf. Theory}, vol. 55, no. 8, pp. 3751-3780, Aug 2009.

\bibitem{tel}
I. E. Telatar, ``Capacity of multi-antenna Gaussian channels," \emph{Eur. Trans. Telecommun.}, vol. 10, no. 6, pp. 585-595, Nov. 1999.

\bibitem{hoch}
B.M. Hochwald and S. ten Brink, ``Achieving Near-Capacity on a Multiple-Antenna Channel,'' \emph{IEEE Trans. Commun.}, vol. 51, no. 3, pp. 389-399, March 2003.

\bibitem{bar}
E. Baccarelli, ``Evaluation of the Reliable Data Rates Supported by Multiple-Antenna Coded Wireless Links for QAM Transmissions,'' \emph{IEEE J. Sel. Areas Commun.}, vol. 19, no. 2, pp. 295-304, Feb. 2001.

\bibitem{paul}
P. J. McCarthy, \emph{Algebraic Extensions of Fields}. New York: Dover Publications, 1991.

\bibitem{paulo}
P. Ribenboim, \emph{Classical Theory of Algebraic Numbers}. 2nd ed. New York: Springer-Verlag, 2001.

\bibitem{neukirch}
J. Neukirch, \emph{Algebraic Number Theory, (Grundlehren der mathematischen Wissenschaften)}, 1st ed.: Springer, 1999.

\bibitem{fin_field}
R. Lidl and H. Niederreiter, \emph{Finite Fields}, in Encyclopedia of Mathematics and Its Application, 2nd ed., vol 20, Cambridge: Cambridge University Press, 1997.



\end{thebibliography}
\end{document}